\documentclass[11pt]{amsart}
\usepackage[utf8]{inputenc}
\usepackage[T1]{fontenc}
\usepackage[american]{babel}
\usepackage{xspace}
\usepackage{hyperref}
\usepackage{amssymb}
\usepackage{amsthm}
\usepackage{amsmath}
\usepackage{amsfonts}
\usepackage{mathrsfs}  

\usepackage{todonotes}
\usepackage{boxedminipage,verbatim,float,caption}
\usepackage{enumerate}

\usepackage{fullpage}

\def\bN{\mathbb{N}}

\def\cA{\mathcal{A}}
\def\cB{\mathcal{B}}

\def\cP{\mathcal{D}}
\def\cE{\mathcal{E}}
\def\cP{\mathcal{F}}

\def\cM{\mathcal{M}}

\def\cP{\mathcal{P}}
\def\cQ{\mathcal{Q}}

\def\ie{\emph{i.e.}\xspace}

\newcommand{\Sub}{\mathsf{Sub}}
\newcommand{\reduce}{\mathsf{reduce}}

\newcommand{\cw}{\mathsf{cw}}

\newtheorem{theorem}{Theorem}[section]
\newtheorem{lemma}[theorem]{Lemma}

\newtheorem{claim}{Claim}[theorem]
\newtheorem{fact}[theorem]{Fact}
\newtheorem{definition}[theorem]{Definition}

\renewcommand{\operatorname}[1]{\mathsf{#1}}

\newcommand{\mult}{\operatorname{mult}}
\newcommand{\bdeg}{\operatorname{bdeg}}
\newcommand{\rdeg}{\operatorname{rdeg}}
\renewcommand{\deg}{\operatorname{deg}}

\newcommand{\aux}{\operatorname{aux}}
\newcommand{\lab}{\operatorname{lab}}
\newcommand{\val}{\operatorname{val}}

\newcommand{\NP}{\textsf{NP}}
\newcommand{\FPT}{\textsf{FPT}\xspace}
\newcommand{\XP}{\textsf{XP}\xspace}
\newcommand{\W}{\textsf{W}}
\newcommand{\ETH}{\textsf{ETH}\xspace}

\begin{document}
	
\title[Hamiltonian Cycle parameterized by Clique-width]{An optimal  XP algorithm for Hamiltonian Cycle on graphs of bounded clique-width}

\author[B.~Bergougnoux]{Benjamin Bergougnoux}
\address{Universit\'e Clermont Auvergne, LIMOS, CNRS, Aubi\`ere, France}
\email{\href{mailto:benjamin.bergougnoux@uca.fr}{\texttt{benjamin.bergougnoux@uca.fr}}}

\author[M.~M.~Kant\'e]{Mamadou Moustapha Kant\'e}
\email{	\href{mailto:mamadou.kante@uca.fr}{\texttt{mamadou.kante@uca.fr}}}

\author{O-joung Kwon}
\address{Department of Mathematics, Incheon National University, Incheon, South Korea}
\email{	\href{mailto:ojoungkwon@gmail.com}{\texttt{ojoungkwon@gmail.com}}}

\thanks{ An extended abstract appeared in 
  Algorithms and Data Structures, WADS 2017~\cite{BergougnouxKK17}.
  B. Bergougnoux and M.M. Kant\'e are supported by French Agency for Research under the GraphEN project (ANR-15-CE40-0009). O. Kwon is supported by the European Research Council	(ERC) under the European Union's Horizon 2020 research and innovation programme (ERC consolidator grant DISTRUCT, agreement No. 648527), and also supported by the National Research Foundation of Korea (NRF) grant funded by the Ministry of Education (No. NRF-2018R1D1A1B07050294).}

\begin{abstract}

In this paper, we prove that, given a clique-width $k$-expression of an $n$-vertex graph,
\textsc{Hamiltonian Cycle} can be solved in time $n^{\mathcal{O}(k)}$. This improves the naive algorithm that runs in time $n^{\mathcal{O}(k^2)}$ by Espelage et al. (WG 2001), 
and it also matches with the lower bound result by Fomin et al. that, unless the Exponential Time Hypothesis fails, there is no algorithm running in time $n^{o(k)}$ (SIAM. J. Computing 2014).

We present a technique of representative sets using two-edge colored multigraphs on $k$ vertices.  
The essential idea is that, for a two-edge colored multigraph, the existence of an Eulerian trail that uses edges with different colors alternately can be
determined by two information:
 the number of colored edges incident with each vertex, and the connectedness of the multigraph.
With this idea, we avoid the bottleneck of the naive algorithm, which stores all the possible multigraphs on $k$ vertices with at most $n$ edges. 
\end{abstract}

\date{\today}
	
\maketitle

	\section{Introduction}\label{sec:intro}
	
	\emph{Tree-width} is one of the most well-studied graph parameters in the graph algorithm community, due to its numerous structural and algorithmic properties (see the survey \cite{Bodlaender98}). 
	A celebrated algorithmic meta-theorem by Courcelle~\cite{Courcelle90} states that every graph problem expressible in monadic second-order logic (MSO$_2$) can be decided in linear time on graphs of bounded tree-width. Among the problems expressible in MSO$_2$, there are some well-known \NP-hard problems such as  \textsc{Minimum Dominating Set}, \textsc{Hamiltonian Cycle}, and \textsc{Graph Coloring}.
	
	Despite the broad interest on tree-width, only sparse graphs can have bounded tree-width. 
	But, on many dense graph classes, some \NP-hard problems admit polynomial-time algorithms, and many of these algorithms can be explained by the boundedness of their \emph{clique-width}.
	Clique-width is a graph parameter that originally emerges from the theory of graph grammars~\cite{CourcelleER93} and the terminology was first introduced by Courcelle and Olariu \cite{CourcelleO00} (see also the book \cite{CourcelleE12}).

	Clique-width is defined in terms of the following graph operations: (1) addition of a single labeled vertex,~ (2) addition of all possible edges between vertices labeled $i$ and those labeled $j$,~ (3) renaming of labels,~ and (4) taking the disjoint union of two labeled graphs. The clique-width of a graph is the minimum number of labels needed to construct it. An expression constructing a graph with at most $k$ labels is called a \emph{(clique-width) $k$-expression}. 
	The modeling power of clique-width is strictly stronger than the modeling power of tree-width. 
	In other words, if a graph class has bounded tree-width, then it has bounded clique-width, but the converse is not true.

	Computing the clique-width of a graph is a problem which has received significant attention over the last decade. 
	Fellows et al.~\cite{FellowsRRS06} showed that computing clique-width is \NP-hard.
	For a fixed $k$, the best known approximation algorithm is due to Hlin\v{e}n\'y and Oum \cite{HlinenyO08}; 
	it computes in time $O(f(k)\cdot n^3)$ a $(2^{k+1}-1)$-expression for an $n$-vertex graph of clique-width at most $k$.

	Courcelle, Makowsky, and Rotics \cite{CourcelleMR00} extended the meta-theorem of
        Courcelle~\cite{Courcelle90} to graphs of bounded clique-width at a cost of a smaller set of problems.  More precisely, they showed that every problem expressible in monadic second
        order logic with formula that does not use edge set quantifications (called MSO$_1$) can be decided in time $O(f(k)\cdot n)$ in any $n$-vertex graph of
        clique-width $k$, provided that a clique-width $k$-expression of it is given.
	
	For some MSO$_1$ problems, clique-width and tree-width have sensibly the same behavior.  Indeed, many problems expressible in MSO$_1$ that admit $2^{O(k)}\cdot n^{O(1)}$-time algorithms parameterized by tree-width have been shown to admit $2^{O(k)}\cdot n^{O(1)}$-time algorithms, when a clique-width $k$-expression is given.  These
        include famous problems like \textsc{Minimum Dominating Set} and \textsc{Minimum Vertex Cover} \cite{BuiXuanTV13,Gurski08,RooijBR09}, or even their connected
        variants and \textsc{Feedback Vertex Set} \cite{BergougnouxK19,BodlaenderCKN15,CyganNPPvRW11}.

	On the other hand, several classical problems, such as \textsc{Max-Cut}, \textsc{Edge Dominating Set (EDS)}, \textsc{Graph Coloring (GC)}, and \textsc{Hamiltonian Cycle (HC)},  are not expressible in MSO$_1$.
These problems are known to be \FPT parameterized by tree-width thanks to Courcelle's theorem or some variants of this latter theorem \cite{ArnborgLS91,BoriePT92,CourcelleM93}. 
They are also known to admit \XP algorithms parameterized by clique-width \cite{EspelageGW01, FominGLS14, KoblerR03}.

	A natural question is whether these problems admit algorithms with running time $f(k)\cdot n^{O(1)}$ given a $k$-expression of the input graph.
	Fomin, Golovach, Lokshtanov, and Saurabh \cite{FominGLS10a} proved the $\W[1]$-hardness of \textsc{EDS}, \textsc{GC}, and \textsc{HC} with clique-width as parameter, 
	which implies that these problems do not admit algorithms with running time $f(k)\cdot n^{O(1)}$, for any function $f$, unless $\W[1]= \FPT$.
	In 2014, the same authors \cite{FominGLS14} have proved that \textsc{Max-Cut} and \textsc{EDS} admit $n^{O(k)}$-time algorithms, and that they do not admit $f(k)\cdot n^{o(k)}$-time algorithms unless \ETH fails.
	In the conclusion of \cite{FominGLS14}, the authors state that \textsc{HC} does not admit $f(k)\cdot n^{o(k)}$-time algorithm under \ETH (they gave the proof in \cite{FominGLSZ19}) and they left an open question of finding an algorithm with running time $f(k)\cdot n^{O(k)}$.
	At the time, the best known running time parameterized by clique-width for \textsc{HC} and \textsc{GC} were respectively $n^{O(k^2)}$ \cite{EspelageGW01}  and $n^{O(2^k)}$ \cite{KoblerR03}.
	Recently, Fomin et al. \cite{FominGLSZ19} provided  a lower bound of $f(k)\cdot n^{2^{o(k)}}$  for \textsc{GC}.  
\bigskip
\paragraph{\bf Our Contribution and Approach} In this paper, we prove that there exists an algorithm solving \textsc{Hamiltonian Cycle} in time $n^{O(k)}$, when a
clique-width $k$-expression is given.  A \emph{Hamiltonian cycle} of a graph $G$ is a cycle containing all the vertices of $G$. The problem \textsc{Hamiltonian Cycle}
asks, given a graph $G$, if $G$ contains a Hamiltonian cycle.  Specifically, we prove the following.

	\begin{theorem}
		There exists an algorithm that, given an $n$-vertex graph $G$ and a $k$-expression of $G$, solves \textsc{Hamiltonian Cycle} in time $ O(n^{2k+5} \cdot 2^{2k (\log_2(k)+ 1)} \cdot k^3  \log_2(n k))$.
	\end{theorem}
	
	Our algorithm is a dynamic programming one whose steps are based on the $k$-labeled graphs $H$ arising in the $k$-expression of $G$.
	Observe that the edges of a Hamiltonian cycle of $G$ which belong to $E(H)$ induce either a Hamiltonian cycle or a collection of vertex-disjoint paths in $G$ covering $H$.
	Consequently, we define a partial solution as a set of edges of $H$ which induces a collection of paths (potentially empty) covering $H$.
	As in \cite{EspelageGW01},  with each partial solution $\cP$, 
	we associate an auxiliary multigraph such that its vertices correspond to the labels of $H$ and each edge $\{i,j \}$ corresponds to a maximal path induced by $\cP$ with end-vertices labeled $i$ and $j$.
	
	Since $H$ is a $k$-labeled graph arising in a $k$-expression of $G$, we have that two vertices $x$ and $y$ with the same label in $H$ have the same neighborhood in $G-E(H)$ (the graph $G$ without the edges of $H$).
	It follows that the endpoints of a path in a partial solution are not important and what matters are the labels of these endpoints.
	As a result, two partial solutions with the same auxiliary multigraph are equivalent, \ie, if one is contained in a Hamiltonian cycle, then the other also.
	From these observations, one easily deduces the $n^{O(k^2)}$-time algorithm, due to Espelage, Gurski, and Wanke~\cite{EspelageGW01}, 
	because there are at most $n$ possible paths between two label classes and thus there are at most $n^{\mathcal{O}(k^2)}$ possible auxiliary graphs.
	
	To obtain our $n^{O(k)}$-time algorithm, we refine the above equivalence relation.
	We define that two partial solutions are equivalent if their auxiliary graphs have the connected components on same vertex sets, and the paths they induce have the same number of endpoints labeled $i$, for each label $i$.
	The motivation of this equivalence relation can be described as follows. 
	Suppose that we have a partial solution $\cP$ and a set of edges $\cQ\subseteq E(G)\setminus E(H)$ so that $\cP\cup \cQ$ forms a Hamiltonian cycle, and
	we consider to make an auxiliary graph of $\cQ$, and identify with the one for $\cP$.
	To distinguish edges obtained from $\cP$ or $\cQ$, we color edges by red if one came from $\cP$ and by blue otherwise.
	Then following the Hamiltonian cycle, we can find an Eulerian trail of this merged auxiliary graph that uses edges of distinct colors alternately.
	But then if $\cP'$ is equivalent to $\cP$, 
	then one can observe that if we replace the red part with the auxiliary graph of $\cP'$, then 
	it also has such an Eulerian trail, and we can show that $\cP'$ can also be completed to a Hamiltonian cycle.
	So, in the algorithm, for each equivalence class, we store one partial solution.
	We define this equivalence relation formally in Section \ref{sec:path-col}.

Since, the number of
        partitions of a $k$-size set is at most $k^k$ and the number of paths induced by a partial solution is always bounded by $n$, the number of non-equivalent partial solutions is then
        bounded by $(2n)^k\cdot k^k$ (the maximum degree of an auxiliary multigraph is at most $2n$ because a loop is counted as two edges).  
        	The running time of our algorithm follows from the maximum number of non-equivalent partial solutions.
        The main effort in the algorithm consists then in updating the equivalence classes of this equivalence relation in terms of operations
        based on the clique-width operations. 

	\medskip
	\paragraph{\bf Overview}
	In Section \ref{sec:prelim}, we give basic definitions and notations concerning (multi)graphs and clique-width. 
	Our notions of partial solutions and of auxiliary
        multigraphs are given in Section \ref{sec:path-col}.  In Section \ref{sec:toolkit}, we prove the equivalence between the existence of Hamiltonian cycles in the
        input graph and the existence of red-blue Eulerian trails in auxiliary multigraphs, and deduce that it is enough to store $(2n)^k\cdot k^{k}$ partial
        solutions at each step of our dynamic programming algorithm.  In Section \ref{sec:algo}, we show how to obtain from the results of Section \ref{sec:toolkit} an
        $n^{O(k)}$-time algorithm for \textsc{Hamiltonian Cycle}.
	In Section \ref{sec:dhc}, we give some intuitions for solving in time $n^{O(k)}$ the problems \textsc{Directed Hamiltonian Cycle} given a $k$-expression.
	We end up with some concluding remarks and open questions in Section \ref{sec:conclusion}.
	
	\section{Preliminaries}\label{sec:prelim}
	
	The size of a set $V$ is denoted by $|V|$, and we write $[V]^2$ to denote the set of all subsets of $V$ of size $2$. 
	We denote by $\bN$ the set of non-negative integers.
	For two sets $A$ and $B$, we let
	\begin{align*}
	A \otimes B  & := \begin{cases}  \emptyset & \textrm{if $A=\emptyset$ or $B=\emptyset$,}\\ 
	\{ X\cup Y \mid X\in A \,\wedge\, Y\in B\}  &
	\textrm{otherwise}. \end{cases} 
	\end{align*}
	For a positive integer $n$, let $[n]:=\{1, 2, \ldots, n\}$.
	
	\medskip
	\paragraph{\bf Graph.} We essentially follow Diestel's book~\cite{Diestel05} for our graph terminology, but we deal only with finite graphs. 
	We distinguish graphs and multigraphs, and for graphs we do not allow to have multiple edges or loops, 
	while we allow them in multigraphs.
	The vertex set of a graph $G$ is denoted by $V(G)$ and its edge set is denoted by $E(G)\subseteq [V(G)]^2$.
	We write $xy$ to denote an edge $\{x,y\}$. 

	Let $G$ be a graph.  For $X\subseteq V(G)$, we denote by $G[X]$ the subgraph of $G$ induced by $X$. 
	For $F\subseteq E(G)$, we write $G_{|F}$ for the subgraph $(V(G),F)$ and $G-F$ for the subgraph $(V(G),E(G)\setminus F)$.
	For an edge $e$ of $G$, we simply write $G-e$ for $G-\{e\}$.
	The \emph{degree} of a vertex $x$, denoted by $\deg_G(x)$, is the number of edges incident with $x$. The set of vertices adjacent to a vertex $v$ is denoted by $N_G(x)$.

		\medskip
	\paragraph{\bf Multigraph.} 
	A \emph{multigraph} is essentially a graph, but we allow multiple edges, \ie, edges incident with the same set of vertices. 
	Formally, a \emph{multigraph} $G$ is a pair $(V(G),E(G))$ of disjoint sets, also called sets of
	vertices and edges, respectively, together with a map $\mult_G:E(G)\to V(G)\cup [V(G)]^2$, which maps every edge to one or two vertices, still called its endpoints.
	The degree of a vertex $x$ in a multigraph $G$, is defined analogously as in graphs, except that each loop is counted twice, and similarly for other notions.  
	If there are exactly $k$ edges $e$ such	that $\mult_G(e)=\{x,y\}$ (or $\mult_G(e)=\{x\}$), then we denote these distinct edges by $\{x,y\}_1,\ldots,\{x,y\}_k$ (or $\{x\}_1,\ldots,\{x\}_k$).

	For two multigraphs $G$ and $H$ on the same vertex set $\{v_1,\ldots,v_k\}$ and with disjoint edge sets, we denote by $G\uplus H$ the multigraph with vertex set $\{v_1,\ldots,v_k\}$, edge set $E(G)\cup E(H)$, and
	\[ \mult_{G\uplus H}(e):=\begin{cases}
		\mult_G(e) & \text{if } e\in E(G),\\
		\mult_H(e) & \text{otherwise}.
	\end{cases} \]
	 If the edges of $G$ and $H$ are colored, then this operation preserves the colors of the edges.

	\medskip
	\paragraph{\bf Walk.} A \emph{walk} of a graph is a sequence of vertices and edges, starting and ending at some vertices, called \emph{end-vertices}, where for every consecutive pair of a vertex $x$ and an edge $e$, $x$ is incident with $e$. 
	The vertices of a walk which are not end-vertices are called \emph{internal} vertices.
	A \emph{trail} of a graph is a walk where each edge is used at most once.
	A walk (or a trail) is \emph{closed} if its two end-vertices are the same. 
	Moreover, when the edges of a graph are colored red or blue, we say that a walk $W=(v_1,e_1,\dots, v_{r-1}, e_{r-1},v_r)$ is a \emph{red-blue walk}, if, for every $i\in\{1,\dots,r-2\}$, the colors of $e_i$ and $e_{i+1}$ are different and the colors of $e_1$ and $e_{r-1}$ are different, when the walk is closed. 
	We adapt all the notations to multigraphs as well.

	For two walks $W_1=(v_1,e_1,\dots,e_{\ell-1},v_\ell)$ and $W_2=(v'_1,e'_1,\dots,e'_{r-1},v'_r)$ such that $v_\ell=v'_1$, the \emph{concatenation} of $W_1$ and $W_2$, denoted by $W_1-W_2$, is the walk $(v_1,e_1,\dots,e_{\ell-1},v_\ell,e'_1,\dots,e'_{r-1},v'_r)$.

	A \emph{path} of a graph is a walk where each vertex is used at most once.
	A \emph{cycle} of a graph is a closed walk where each vertex other than the end-vertices is used at most once.
		An \emph{(closed) Eulerian trail} in a graph $G$ is a closed trail containing all the edges of $G$.
		In particular, if the edges of a graph are colored red or blue, 
		then a red-blue Eulerian trail is an Eulerian trail that is a red-blue walk.

	\medskip
	\paragraph{\bf Clique-width.}
	 A \emph{$k$-labeled graph} is a pair $(G,\lab_G)$ of a graph $G$  and a function $\lab_G$ from $V(G)$ to $[k]$, called the \emph{labeling
	function}.
	We denote by $\lab_G^{-1}(i)$ the set of vertices in $G$ with label $i$.
	The notion of clique-width is defined by  Courcelle and Olariu \cite{CourcelleO00}
	 and is based on the following operations:
	\begin{itemize}
		\item creating a graph, denoted by $i(x)$, with a single vertex $x$ labeled with $i\in [k]$;

		\item for a labeled graph $G$ and distinct labels $i,j \in [k]$, relabeling the vertices of $G$ with label $i$ to $j$, denoted by $\rho_{i\to j}(G)$; 
		\item for a labeled graph $G$ and distinct labels $i,j\in [k]$, adding all the non-existent edges between vertices with label $i$ and vertices with label $j$, denoted by $\eta_{i,j}(G)$;
		\item taking the disjoint union of two labeled graphs $G$ and $H$, denoted by $G\oplus H$,  with 
		\[ \lab_{G\oplus H}(v):=
		\begin{cases}
		\lab_G(v) & \text{if } x\in V(G),\\
		\lab_H(v) & \text{otherwise.}
		\end{cases} \]

	\end{itemize}
	
	A \emph{clique-width $k$-expression}, or shortly a \emph{$k$-expression}, is a finite well-formed term built with the four operations above and using at most $k$ labels. 
	Each $k$-expression $\phi$ evaluates into a $k$-labeled graph $(\val(\phi),\lab_{\val(\phi)})$. 
	The \emph{clique-width} of a graph $G$, denoted by $\cw{(G)}$, is the minimum $k$ such that $G$ is isomorphic to $\val(\phi)$ for some $k$-expression $\phi$. 
	We can assume without loss of generality that any $k$-expression defining a graph $G$ uses $O(n)$ disjoint union operations and $O(nk^2)$ unary operations \cite{CourcelleO00}.
	
	It is worth noticing, from the recursive definition of $k$-expressions, that one can compute in time linear in $|\phi|$ the labeling function $\lab_{\val(\phi)}$ of $\val(\phi)$, and hence we
        will always assume that it is given. 
	
	For example, the cycle $abcdea$ of length $5$ can be constructed using the $3$-expression represented as a tree-structure in Figure~\ref{fig:c5expression}.

	\begin{figure}[h]
		\center
		\begin{tikzpicture}[level distance=8mm]
		\tikzstyle{level 1}=[sibling distance=24mm] 
		\tikzstyle{level 6}=[sibling distance=12mm] 
		\node {$\eta_{1,3}$}
		child {node {$\oplus$} 
			child {node {$\rho_{3\rightarrow 2}$}
				child {node {$\eta_{2,3}$}
					child {node {$\oplus$}
						child {node {$ \eta_{1,2}$}
							child {node {$\oplus$}
								child {node {$1(a)$}}
								child {node {$2(b)$}}}}
						child {node {$ \eta_{1,3}$}
							child {node {$\oplus$}
								child {node {$3(c)$}}
								child {node {$1(d)$}}}	
					}}
			}}
			child {node {$3(e)$}}};
		\end{tikzpicture}
		\caption{An irredundant $3$-expression of $C_5$.} 	\label{fig:c5expression}
	\end{figure}
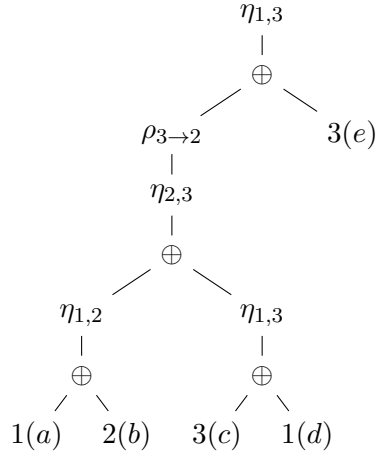
	
	The set of subexpressions of a $k$-expression $\phi$, denoted by $\Sub(\phi)$, is defined by the following induction: 
	\[ \Sub(\phi):=\begin{cases}
	\{\phi\} \text{ if } \phi:=i(x) \text{ with } i\in [k],\\
	\{\phi\}\cup \Sub(\phi')\cup \Sub(\phi^\star)\text{ if }\phi=\phi'\oplus\ \phi^\star,\\
	\{\phi\}\cup \Sub(\phi')\text{ if }\phi=f(\phi')\text{ with }f\in\{\rho_{i\rightarrow j},\eta_{i,j}\mid i,j\in[k]\}.
	\end{cases} \]
	We say that a $k$-labeled graph $(H,\lab_H)$ arises in a $k$-expression $\phi$ if $H=\val(\phi')$ and $\lab_H=\lab_{\val(\phi')}$, for some $\phi'\in\Sub(\phi)$.
	
	A $k$-expression $\phi$ is called \emph{irredundant} if, for every subexpression $\eta_{i,j}(\phi')$ of $\phi$, there are no edges in $\val(\phi')$ between the vertices labeled $i$ and the vertices labeled $j$.
	Courcelle and Olariu~\cite{CourcelleO00} proved that given a clique-width $k$-expression, it can be transformed into an irredundant $k$-expression in linear time. 
	The following useful properties of an irredundant $k$-expression will be used in Section~\ref{sec:toolkit}.

	\begin{lemma}\label{lem:cwproperties}
		Let $H$ be a $k$-labeled graph arising in an irredundant $k$-expression $\phi$ of a graph $G$.
		For all $u,v\in V(H)$ with $\lab_H(u)=i$ and $\lab_H(v)=j$, we have the following.
		\begin{enumerate}
			\item If $i=j$, then $ N_G(u)\setminus V(H) = N_G(v)\setminus V(H)$.
			\item If $uv\in E(G)\setminus E(H)$, then $i\neq j$ and, for all $x,y\in V(H)$ with $\lab_H(x)=i$ and $\lab_H(y)=j$, we have $xy\in E(G)\setminus E(H)$.
		\end{enumerate}		
	\end{lemma}
	\begin{proof}
		(1) Assume that  $i=j$. 
		Let $\phi'$ be the subexpression of $\phi$ such that $H= \val(\phi')$.
		As $u$ and $v$ have the same label in $H$, 
		in every subexpression of $\phi$ having $\phi'$ as a subexpression, $u$ and $v$ have the same label. 
		Since edges are added only through the operation $\eta_{a,b}$, we conclude that $ N_G(u)\cap (V(G)\setminus V(H)) = N_G(v)\cap (V(G)\setminus V(H))$.
		
		\medskip
		(2) 
		Assume that $uv\in E(G)\setminus E(H)$. Then, we have $i\neq j$ because the operation $\eta_{a,b}$ adds edges only between vertices with distinct labels.
		Let $\phi'$ be the minimal (size wise) subexpression of $\phi$ such that $uv\in E(\val(\phi'))$.
		It follows that $\phi'=\eta_{a,b}(\phi^\star)$, with $\phi^\star\in\Sub(\phi)$, $\lab_{\val(\phi')}(u)=a$ and $\lab_{\val(\phi')}(v)=b$.
		Let $D:=\val(\phi^\star)$.
		Observe that we have $V(H)\subseteq V(D)$ and $E(H)\subseteq E(D)$. 
		Moreover, all vertices labeled $i$ in $H$ are labeled $a$ in $D$ and those labeled $j$ in $H$ are labeled $b$ in $D$.
		Since $\phi$ is irredundant, there are no edges in $D$ between a vertex labeled $a$ and one labeled $b$.
		Thus, for all vertices $x\in \lab_H^{-1}(i)$ and $y\in \lab_H^{-1}(j)$, we have $xy \notin E(H)$ and $xy\in E(G)$.
	\end{proof}

	\section{Partial solutions and auxiliary graphs}\label{sec:path-col}
	
	Let $G$ be a graph and $(H,\lab_H)$ be a $k$-labeled graph such that $H$ is a subgraph of $G$.
	
	A \emph{partial solution} of $H$ is a set of edges $\cP\subseteq E(H)$ such that $H_{|\cP}$ is a union of vertex-disjoint paths, \ie, $H_{|\cP}$ is acyclic and, for every vertex $v\in V(H)$, the degree of $v$ in $H_{|\cP}$ is at most two.
	We denote by $\Pi(H)$ the set of all partial solutions of $H$. We say that a path $P$ in $H_{|\cP}$ is \emph{maximal} if the degree of its end-vertices in $H_{|\cP}$ is at most one; in other words, there is no path $P'$ in $H_{|\cP}$ such that $V(P)\subsetneq V(P')$. Observe that an isolated vertex in $H_{|\cP}$ is considered as a maximal path.
	
	A \emph{complement solution} of $H$ is a subset $\cQ$ of $E(G)\setminus E(H)$ such that $G_{|\cQ}$ is a union of vertex-disjoint paths with end-vertices in $V(H)$; in particular, for every vertex $v$ in $V(G)\setminus V(H)$, the degree of $v$ in $G_{|\cQ}$ is two.
	We denote by $\overline{\Pi}(H)$ the set of all complement solutions of $H$.
	A path $P$ in $G$ with at least $2$ vertices is an \emph{$H$-path} if the end-vertices of $P$ are in $V(H)$ and the internal vertices of $P$ are in $V(G)\setminus V(H)$. 
	By definition, isolated vertices in $V(H)$ are not $H$-paths.
	Observe that, for a complement solution $\cQ$, we can decompose each maximal path $Q$ of $G_{|\cQ}$ with at least $2$ vertices into $H$-paths (not necessarily one). 
	
	Examples of a partial solution and a complement solution are given in Figure \ref{fig:partialsolution}.
	Note that if $G$ has a Hamiltonian cycle $C$ and $E(C)\not\subseteq E(H)$, then $E(C)\cap E(H)$ is a partial solution and $E(C)\cap (E(G)\setminus E(H))$ is a complement solution. 
	We say that a partial solution $\cP$ and a complement solution $\cQ$ form a Hamiltonian cycle 
	if $(V(G), \cP\cup \cQ)$ is a cycle containing all the vertices of $G$.
	
	\begin{figure}[h]
		\includegraphics[width=0.8 \columnwidth]{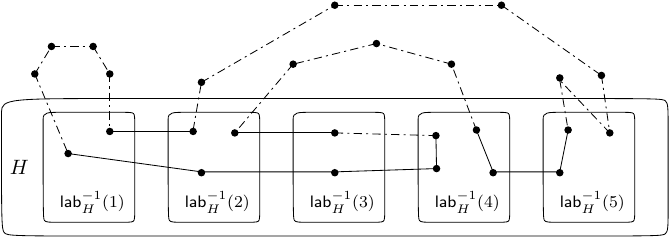}
		\caption{Examples of a partial solution $\cP$ (solid lines) and a complement solution $\cQ$ (dashed lines) forming a Hamiltonian cycle. Observe that $H_{|\cP}$ contains 5 maximal paths and $G_{|\cQ}$ contains $5$ $H$-paths (and only $4$ maximal paths).}
		\label{fig:partialsolution}
	\end{figure}
	\medskip
	\paragraph{\bf Auxiliary Multigraph} For $\cP\in \Pi(H)\cup \overline{\Pi}(H)$ and $i,j\in [k]$, 
	we define $\ell_{ij}$ and $\ell_i$ as follows. 
	\begin{itemize}
		\item If $\cP$ is a partial solution of $H$, then $\ell_{ij}$ is the number of maximal paths in $H_{|\cP}$ with end-vertices labeled $i$ and $j$, and $\ell_i$ is the number of maximal paths in $H_{|\cP}$ with both end-vertices labeled $i$. 
		\item If $\cP$ is a complement solution of $H$, then $\ell_{ij}$ is the number of $H$-paths in $G_{|\cP}$ with end-vertices labeled $i$ and $j$, and $\ell_i$ is the number of $H$-paths in $G_{|\cP}$ with both end-vertices labeled $i$. 
	\end{itemize}
	Now, 
	we define the auxiliary multigraph of $\cP$, denoted by $\aux_H(\cP)$,  as the multigraph with vertex set $\{v_1,\dots,v_k\}$ and  edge set
	\[ \bigcup_{ \substack{i,j\in [k] \\ i\neq j}} \{\{v_i,v_j\}_1,\ldots,\{v_i,v_j\}_{\ell_{ij}}\} \cup \bigcup\limits_{i\in [k]} \{ \{v_i\}_1,\ldots,\{v_i\}_{\ell_{i}} \}. \]
	Moreover, if $\cP$ is a partial solution of $H$, then we color all the edges of $\aux_H(\cP)$ in red, 
	and if $\cP$ is a complement solution, then we color the edges of $\aux_H(\cP)$ in blue.
	An example of an auxiliary multigraph is given in Figure~\ref{fig:aux}.
	\begin{figure}[h]
		\includegraphics[width=0.7\columnwidth]{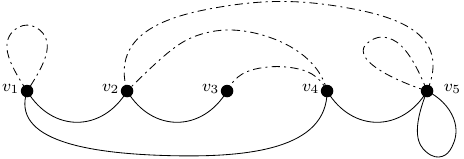}
		\caption{The union $G_1\uplus G_2$ of auxiliary multigraphs $G_1$ and $G_2$ associated with the partial solution (solid lines) and the complement solution (dashed lines) represented in Figure \ref{fig:partialsolution}.}
		\label{fig:aux}
	\end{figure}
	
	\section{Relations between Hamiltonian cycles and Eulerian trails}\label{sec:toolkit}
	
	Let $G$ be an $n$-vertex graph and $\phi$ be an irredundant $k$-expression of $G$.
	Let $H$ be a $k$-labeled graph arising in the $k$-expression $\phi$.
	Observe that $H$ is a subgraph of $G$ (disregarding the labels).
	This section is dedicated to prove the properties of the following relation between partial solutions of $H$ based on the degree sequence and the connected components of their auxiliary multigraphs.
	
	\begin{definition}\label{def:equiv}
		Let $\cP_1, \cP_2\in \Pi(H)$.
		We write $\cP_1 \simeq \cP_2$ if $\aux_H(\cP_1)$ and $\aux_H(\cP_2)$ have the same set of connected components and 
		for each vertex $v$ in $\{v_1,\dots,v_k\}$, $\deg_{\aux_H(\cP_1)}(v)=\deg_{\aux_H(\cP_2)}(v)$.

	Observe that $\simeq$ is an equivalence relation.
		For a set $\cA\subseteq\Pi(H)$, we define $\reduce_H(\cA)$ as the operation which returns a set containing one element of each equivalence class of $\cA/\simeq$.
	\end{definition}

	The main idea of our algorithm is to call $\reduce_H$ at each step of our dynamic programming algorithm in order to keep the size of a set of partial solutions manipulated small, \ie, bounded by $n^{O(k)}$.
	The running time of our algorithm follows mostly from the following lemma.
	
	\begin{lemma}\label{lem:nbequivalence}
		For every $\cA\subseteq \Pi(H)$, we have $|\reduce_H(\cA)|\leq n^{k}\cdot 2^{k(\log_2(k)+ 1)}$ and we can moreover compute $\reduce_H(\cA)$ in time
          $O(|\cA| \cdot nk^2 \log_2(n k))$.
	\end{lemma}
	\begin{proof}
		To prove that $\reduce_H(\cA)\leq n^{k}\cdot 2^{k(\log_2(k)+ 1)}$ , it is enough to bound the number of equivalence classes of $\Pi(H)/\simeq$.
		
		We claim that, for every $\cP\in \Pi(H)$, we have $\sum_{i\in [k]}\deg_{\aux_H(\cP)}(v_i) \leq 2 |V(H)|$.
		First observe that $\sum_{i\in [k]}\deg_{\aux_H(\cP)}(v_i) = 2 |V(H)|$ when $\cP=\emptyset$, since each isolated vertex in $H_{|\cP}$ gives a loop in $\aux_H(\cP)$. 
		Moreover, when $\cP$ contains an edge, removing an edge from a partial solution $\cP$ of $H$ increases $\sum_{i\in [k]}\deg_{\aux_H(\cP)}(v_i)$ by two; indeed, this edge removal splits a maximal path of $H_{|\cP}$ into two maximal paths.
		Therefore, any partial solution $\cP$ satisfies that $\sum_{i\in [k]}\deg_{\aux_H(\cP)}(v_i) \leq 2 |V(H)|$;
		in particular each vertex of $\aux_H(\cP)$ has degree at most $2|V(H)|$.
		As $\aux_H(\cP)$ contains $k$ vertices, we deduce that there are at most $(2|V(H)|)^k\leq (2n)^k$ possible degree sequences for $\aux_H(\cP)$.

		Since the number of partitions of $\{v_1,\dots,v_k\}$ is bounded by $2^{k \log_2 k}$. 
		We conclude that $\simeq$ partitions $\Pi(H)$ into at most $n^{k} \cdot2^{k (\log_2 k+ 1)}$ equivalences classes.
		
		It remains to prove that we can compute $\reduce_H(\cA)$ in time $O(|\cA| \cdot n k \log_2(nk))$.
		First observe that, for every $\cP\in \Pi(H)$, we can compute $\aux_H(\cP)$ in time $O(nk)$.
		Moreover, we can also compute the degree sequence of $\aux_H(\cP)$ and the connected components of $\aux_H(\cP)$ in time $O(nk)$.
		Thus, by using the right data structures, we can decide whether $\cP_1\simeq\cP_2$ in time $O(n k)$.
		Furthermore, by using a self-balancing binary search tree, we can compute $\reduce_H(\cA)$ in time  $O(|\cA| \cdot n k \log_2(|\reduce_H(\cA)|))$.
		Since $\log_2(|\reduce_H(\cA)|)\leq k \log_2(2nk)$, we conclude that $\reduce_H(\cA)$ is computable in time $O(|\cA| \cdot n k^2 \log_2(nk))$.
	\end{proof}
	
	The rest of this section is dedicated to prove that, for a set of partial solutions $\cA$ of $H$, the set $\reduce_H(\cA)$ is equivalent to $\cA$, \ie, if $\cA$ contains a partial solution that forms a Hamiltonian cycle with a complement solution, then $\reduce_H(\cA)$ also.
	Our results are based on a kind of equivalence between Hamiltonian cycles and red-blue  Eulerian trails.
	The following observation is one direction of this equivalence.  
	
	\begin{lemma}\label{lem:easyway}
		If $\cP\in \Pi(H)$ and $\cQ\in \overline{\Pi}(H)$ form a Hamiltonian cycle, then the multigraph $\aux_H(\cP)\uplus \aux_H(\cQ)$ admits a red-blue  Eulerian trail.
	\end{lemma}
	\begin{proof}
		Suppose that $\cP\in \Pi(H)$ and $\cQ\in \overline{\Pi}(H)$ form a Hamiltonian cycle $C$.
		Let $M:=\aux_H(\cP)\uplus \aux_H(\cQ)$.
		From the definitions of a partial solution and of a complement solution, 
		there is a sequence $(P_1, Q_1, \ldots, P_{\ell}, Q_{\ell})$ of paths in $\cP$ and $\cQ$ such that  
		\begin{itemize}
		\item $P_1, P_2, \ldots, P_{\ell}$ are all the maximal paths in $H_{|\cP}$,
		\item $Q_1, Q_2, \ldots, Q_{\ell}$ are all the $H$-paths in $G_{|\cQ}$, 
		\item $P_1, Q_1, \ldots, P_{\ell}, Q_{\ell}$ appear in $C$ in this order,  
		\item for each $x\in[\ell]$, the first end-vertex of $P_x$ is the last end-vertex of $Q_{x-1}$ and the last end-vertex of $P_x$ is the first end-vertex of $Q_x$ (indices are considered modulo $\ell$). 
		\end{itemize}
		
		Observe that each maximal path $P_x$ in $H_{|\cP}$ with end-vertices labeled $i$ and $j$ is associated with a red edge in $M$, say $e_x$ with $\mult_M(e_x)=\{v_i,v_j\}$ if $i\neq j$ or $\mult_M(e_x)=\{v_i\}$ if $i=j$ such that the edges $e_1,\dots,e_\ell$ are pairwise distinct and $E(\aux_H(\cP))=\{e_1,\dots,e_\ell\}$.
		Similarly, each $H$-path $Q_y$ of $G_{|\cQ}$ with end-vertices labeled $i$ and $j$ is associated with a blue edge $f_y$ in $M$ with $\mult_M(f_y)=\{v_i,v_j\}$ if $i\neq j$ or $\mult_M(f_y)=\{v_i\}$ if $i=j$ such that the edges $f_1,\dots,f_\ell$ are pairwise distinct and $E(\aux_H(\cQ))=\{f_1,\dots,f_\ell\}$.
		It is not difficult to see that 
		$(e_1, f_1, \ldots, e_{\ell}, f_{\ell})$ is a red-blue  Eulerian trail of $\aux_H(\cP)\uplus \aux_H(\cQ)$.
	\end{proof}

	Next, we prove the other direction. We use the properties of an irredundant $k$-expression described in Lemma~\ref{lem:cwproperties}.  
	\begin{lemma}\label{lem:EulerianHamiltonian}
		Let $\cP\in \Pi(H)$. 
		If there exists a complement solution $\cQ$ of $H$ such that $\aux_H(\cP)\uplus \aux_H(\cQ)$ admits a red-blue  Eulerian trail, then there exists $\cQ^\star\in \overline{\Pi}(H)$ such that $\cP$ and $\cQ^\star$ form a Hamiltonian cycle.
	\end{lemma}
	\begin{proof}
		Let $T=(v_{a_1},r_1,v_{c_1},b_1,v_{a_{2}},r_2,v_{c_2},\cdots, v_{a_\ell},r_\ell,v_{c_\ell},b_\ell,v_{a_1})$ be a red-blue  Eulerian trail of $\aux_H(\cP)\uplus \aux_H(\cQ)$ with $r_1,\dots,r_\ell\in E(\aux_H(\cP))$ and $b_1,\dots,b_\ell \in E(\aux_H(\cQ))$. 
		In the following, the indexes are modulo $\ell$.
		
		For each $i\in [\ell]$, we associate $r_i$ with a maximal path $P_i$ of $H_{|\cP}$ with end-vertices labeled $a_i$ and $c_{i}$  and we associate $b_i$ with an $H$-path $Q_i$ of $G_{|\cQ}$ with end-vertices labeled $c_{i}$ and $a_{i+1}$, such that $P_1, \ldots, P_{\ell}, Q_1, \ldots, Q_{\ell}$ are all pairwise distinct.
		
		For every $i\in [\ell]$, we construct from $Q_i$ an $H$-path $Q_i^\star$ of $G$ by doing the following.
		Let $u,v$ be respectively the last end-vertex of $P_i$ and the first end-vertex of $P_{i+1}$.
		Observe that $u$ and the first vertex of $Q_i$ are labeled $c_i$, and $v$ and the last vertex of $Q_i$ are labeled $a_{i+1}$.
		We distinguish two cases:
		\begin{itemize}
			\item Suppose that $Q_i=(x,xy,y)$, \ie, $Q_i$ uses only one edge.
			Since $\cQ$ is a complement solution of $H$, we have $xy\in E(G)\setminus E(H)$.
			By Lemma \ref{lem:cwproperties}, we have $uv\in E(G)\setminus E(H)$.
			In this case, we define $Q^\star_i=(u,uv,v)$.
			
			\item Assume now that $Q_i=(x,xy,y,\dots,w,wz,z)$ with $w,y\in V(G)\setminus V(H)$ (possibly, $w=y$).
			Since $x$ and $u$ have the same label in $H$, by Lemma \ref{lem:cwproperties}, we have $N_G(x)\setminus V(H) = N_G(u)\setminus V(H)$.
			Hence, $u$ is also adjacent to $y$. Symmetrically, we have $v$ is adjacent to $w$.
			In this case, we define $Q_i^\star=(u,uy,y,\dots,w,wv,v)$, \ie, the path with the same internal vertices as $Q_i$ and with end-vertices $u$ and $v$.
		\end{itemize}
		
		In both cases, we end up with a path $Q_i^\star$ that uses the same internal vertices as $Q_i$ and whose end-vertices are the last vertex of $P_i$ and the first vertex of $P_{i+1}$. 
		We conclude that 
		\[P_1 -  Q_1^\star - \cdots - P_\ell - Q_\ell^\star\] is a Hamiltonian cycle.
		
		Let $\cQ^\star$ be the set of edges used by the paths $Q_1^\star,\dots,Q_\ell^\star$.
		By construction, we have  $\cQ^\star\subseteq E(G)\setminus E(H)$, and thus $\cQ^\star \in \overline{\Pi}(H)$.
		Observe that, for every $i\in [\ell]$, the labels of the end-vertices of $Q_i^\star$ are the same as those of $Q_i$.
		Consequently, we have $\aux_H(\cQ^\star)=\aux_H(\cQ)$.
	\end{proof}

	It is well known that a connected multigraph contains an Eulerian trail if and only if every vertex has even degree. 
	As an extension, Kotzig~\cite{Kotzig68} proved that a connected two-edge colored graph (without loops and multiple edges) contains a red-blue  Eulerian trail if and only if each vertex is incident with the same number of edges for both colors. This result can be easily generalized to multigraphs by replacing red edge with a path of length $3$ whose colors are red, blue, red in the order, and replacing blue edge with a path of length $3$ whose colors are blue, red, blue in the order.
	For the completeness of our paper, we add its proof.

	Let $G$ be a multigraph whose edges are colored red or blue, and let $R$ and $B$ be respectively the set of red and blue edges. For a vertex $v\in V(G)$, we let 
        $\rdeg_G(v)$ and $\bdeg_G(v)$ be respectively the degree of $v$ in $G_{|R}$ and $G_{|B}$. Recall that a loop is counted twice in the degree of a vertex.
	\begin{lemma}[Kotzig~\cite{Kotzig68}]\label{lem:euleriantrail}
		Let $G$ be a connected multigraph whose edges are colored red or blue. 
		Then $G$ has a red-blue Eulerian trail if and only if, for every vertex $v$, $\bdeg_G(v)=\rdeg_G(v)$. 
	\end{lemma}
	\begin{proof}
	One can easily check that if $G$ has a red-blue  Eulerian trail, then 
		for every vertex $v$, $\bdeg_G(v)=\rdeg_G(v)$. 
		Indeed, if
          $T=(v_1,e_1,v_2,\ldots,v_\ell,e_\ell,v_1)$ is a red-blue Eulerian trail, then, for $2\leq i \leq \ell$, $e_{i-1}$ and $e_i$ have different colors, and $e_1$ and $e_\ell$ have different colors,
          we can conclude that the blue edges incident with a vertex $v$ are in $1$-to-$1$ correspondence with the red edges incident with $v$ (by counting twice the loops).

	Let us now prove the other direction.
	Let $T:=(v_1,e_1,v_2,e_2, \cdots, v_i, e_i, v_{i+1})$ be a longest red-blue  trail. 
	We may assume that $e_1$ is colored red.
	First observe that $v_1=v_{i+1}$. Otherwise, $\bdeg_T(v_1)+1=\rdeg_T(v_1)$ and thus, there is a blue edge in $E(G)\setminus E(T)$ incident with $v_1$. 
	So, we can extend $T$ by adding this edge, a contradiction.
	Thus, $v_1=v_{i+1}$.
	
	Next we show that $e_i$ is colored blue. Suppose $e_i$ is colored red.
	Then $\bdeg_T(v_1)+2=\rdeg_T(v_1)$ and thus, there is a blue edge in $E(G)\setminus E(T)$ incident with $v_1$.
	So, we can extend $T$ by adding this edge, a contradiction.
	Thus, $e_i$ is colored red, and it implies that $T$ is a closed red-blue trail.
	It means that $T$ can be considered as a closed red-blue trail starting from any vertex of $T$ and following the same order or the reverse order of $T$.
	
	Now, we show that $V(G)=V(T)$. Suppose $V(G)\setminus V(T)$ is non-empty. Since $G$ is connected, 
	there is an edge $vw$ with $v\in V(T)$ and $w\in V(G)\setminus V(T)$. If $vw$ is a red edge, 
	then starting from this edge and following $T$ from a blue edge incident with $v$, we can find a red-blue trail longer than $T$, a contradiction.
	The same argument holds when $vw$ is a blue edge.
	Therefore, we have that $V(G)=V(T)$.
	By a similar argument, one can show that $E(G)=E(T)$; if there is an edge $vw$ in $E(G)\setminus E(T)$, we can extend $T$ by putting $vw$ at the beginning.
	So, $E(G)=E(T)$.
	
	We conclude that $T$ is a red-blue Eulerian trail, as required.
	\end{proof}

	In order to prove the correctness of our algorithm, we need the following relation between subsets of partial solutions.

	\begin{definition}\label{def:bigequiv}
		Let $\cA$ and $\cB$ be two subsets of $\Pi(H)$.
		We write $\cA\lesssim_H\cB$ if, for every multigraph $\cM$ whose edges are colored blue, whenever there exists $\cP_1\in \cB$ such that $\aux_H(\cP_1)\uplus \cM$ admits a red-blue  Eulerian trail, there exists $\cP_2\in \cA$ such that $\aux_H(\cP_2)\uplus\cM$ admits a red-blue  Eulerian trail.
	\end{definition}

	The main idea of our algorithm for \textsc{Hamiltonian Cycle}, is to compute, for every $k$-labeled graph $H$ arising in $\phi$, a set $\cA\subseteq\Pi(H)$ of small size such that $\cA\lesssim_H \Pi(H)$. Indeed, by Lemmas \ref{lem:easyway} and \ref{lem:EulerianHamiltonian}, $\cA\lesssim_H \Pi(H)$ implies that if there exist $\cP\in \Pi(H)$ and $\cQ\in\overline{\Pi}(H)$ such that $\cP$ and $\cQ$ form a Hamiltonian cycle, then there exist $\cP^\star\in\cA$ and $\cQ^\star\in \overline{\Pi}(H)$ such that $\cP^\star$ and $\cQ^\star$ form a Hamiltonian cycle.
	The following lemma is the key of our algorithm.

	\begin{lemma}\label{lem:reduce1}
		Let $\cA\subseteq \Pi(H)$. 
		Then $\reduce_H(\cA)\lesssim_H \cA$.
	\end{lemma}
	\begin{proof}
		Let $\cP\in \cA$ and $\cM$ be a multigraph whose edges are colored blue such that $\aux_H(\cP)\uplus\cM$ admits a red-blue  Eulerian trail.
		By definition, $\reduce_H(\cA)$ contains a partial solution $\cP^\star$ such that $\cP\simeq \cP^\star$.
		As $\aux_H(\cP)\uplus\cM$ contains a red-blue Eulerian trail, by Lemma~\ref{lem:euleriantrail}, we have that 
		\begin{itemize}
		\item $\aux_H(\cP)\uplus\cM$ is connected and 
		\item for every $i\in[k]$,  $\deg_{\aux_H(\cP)}(v_i)=\deg_{\cM}(v_i)$.
                \end{itemize}
		Since $\aux_H(\cP)$ has the same set of connected components as $\aux_H(\cP^\star)$, the multigraph $\aux_H(\cP^\star)\uplus\cM$ is also connected.
		Moreover, for every $i\in[k]$, we have 
		\[\deg_{\aux_H(\cP)}(v_i)=\deg_{\aux_H(\cP^\star)}(v_i)=\deg_{\cM}(v_i).\]
		By Lemma \ref{lem:euleriantrail}, we conclude that $\aux_H(\cP^\star)\uplus\cM$ admits a red-blue Eulerian trail.
		
		Thus, for every $\cP\in\cA$ and multigraph $\cM$ with blue edges such that $\aux_H(\cP)\uplus\cM$ admits a red-blue  Eulerian trail, there exists $\cP^\star\in \reduce_H(\cA)$ such that $\aux_H(\cP^\star)\uplus\cM$ admits a red-blue  Eulerian trail. Hence, we have $\reduce_H(\cA)\lesssim_H \cA$.	
	\end{proof}

	\begin{lemma}\label{lem:reduce}
		Let $\cA, \cB\subseteq \Pi(H)$. 
		If $\cA\lesssim_H \cB$, then $\reduce_H(\cA)\lesssim_H \cB$.
	\end{lemma}
	\begin{proof}	
          One easily checks that $\lesssim_H$ is a transitive relation. Now, 
		assuming that $\cA\lesssim_H\cB$, we have $\reduce_H(\cA)\lesssim \cB$ because $\reduce_H(\cA)\lesssim_H \cA$ by Lemma \ref{lem:reduce1}.	
	\end{proof}

	\section{Hamiltonian Cycle problem}\label{sec:algo}
	
	In this section, we present our algorithm solving \textsc{Hamiltonian Cycle}.
	Our algorithm computes recursively, for every $k$-labeled graph $H$ arising in the $k$-expression of $G$, a set $\cA_H$ such that $\cA_H\lesssim_H\Pi(H)$ and $|\cA_H|\leq n^{k}\cdot 2^{k (\log_2(k)+ 1)}$.
	In order to prove the correctness of our algorithm, we need the following lemmas which prove that the operations we use to compute sets of partial solutions preserve the relation $\lesssim_H$.
	
	\begin{lemma}\label{lem:relabeling}
		Let $H=\rho_{i\rightarrow j}(D)$. 
		If $\cA_D\lesssim_D \Pi(D)$, then $\cA_D\lesssim_H \Pi(H)$.
	\end{lemma}
	\begin{proof}
		First, observe that $H$ has the same set of vertices and edges as $D$. Thus, we have $\Pi(H)=\Pi(D)$ and $\overline{\Pi}(H)=\overline{\Pi}(D)$.
		Suppose that $\cA_D\lesssim_D \Pi(D)$.
		
		Let $\cP\in \Pi(H)$ and $\cM$ be a multigraph whose edges are colored blue such that $\aux_H(\cP)\uplus\cM$ contains a red-blue Eulerian trail $T$. 
		To prove the lemma, it is sufficient to prove that there exists $\cP^\star\in\cA_D$ such that $\aux_H(\cP^\star)\uplus\cM$ contains a red-blue Eulerian trail.

		Let $f$ be a bijective function such that 
		\begin{itemize}
		\item for every edge $e$ of $\aux_D(\cP)$ with endpoints $v_\ell$ and $v_i$, for some $\ell$, $f(e)$ is an edge of $\aux_H(\cP)$ with endpoints $v_\ell$ and $v_j$, and 
		\item for every loop $e$ with endpoint $v_i$, $f(e)$ is a loop of $\aux_H(\cP)$ with endpoint $v_j$. 
		\end{itemize}
		By construction of $\aux_D(\cP)$ and $\aux_H(\cP)$, such a function exists.
		
		We construct the multigraph $\cM'$ from $\cM$ and $T$ by successively doing the following: 
		\begin{itemize}
			\item For every edge $e$ of $\aux_D(\cP)$ with endpoints $v_\ell$ and $v_i$, take the subwalk $W=(v_\ell,f(e),v_j,e_a,v_a)$ of $T$. Replace $e_a$ in $\cM$ by an edge $e_a'$ with endpoints $v_i$ and $v_a$. 
			\item For every loop $e$ with endpoint $v_i$ in $\aux_D(\cP)$, take the subwalk $W=(v_a,e_a,v_j,f(e),v_j,e_b,v_b)$ of $T$. Replace $e_a$ (respectively $e_b$) in $\cM$ by an edge $e_a'$ (resp. $e_b'$) with endpoints $v_a$ and $v_i$ (resp. $v_i$ and $v_b$).
		\end{itemize}
		
		By construction, one can construct from $T$ a red-blue Eulerian trail of $\aux_D(\cP)\uplus\cM'$.
		Since $\cA_D\lesssim_D \Pi(D)$, there exists $\cP^\star\in\cA_D$ such that $\aux_D(\cP^\star)\uplus\cM'$ contains a red-blue Eulerian trail.
		Observe that $\aux_H(\cP)$ (respectively $\cM$) is obtained from $\aux_D(\cP^\star)$ (resp. $\cM'$) by replacing each edge associated with $\{v_i,v_k\}$ or $\{v_i\}$ in $\aux_D(\cP^\star)$ (resp. $\cM'$) with an edge associated with $\{v_j,v_k\}$ or $\{v_j\}$ respectively. 
		We conclude that $\aux_H(\cP^\star)\uplus\cM$ admits a red-blue Eulerian trail.
	\end{proof}
	
	\begin{lemma}\label{lem:disjointunion}
		Let $H=D\oplus F$.
		If $\cA_{D}\lesssim_D \Pi(D)$ and $\cA_{F}\lesssim_F \Pi(F)$, then $(\cA_{D}\otimes \cA_{F})\lesssim_H \Pi(H)$.
	\end{lemma}
	\begin{proof}
		Observe that $V(D)$ and $V(F)$ are disjoint.
		Consequently, we have $\Pi(H)=\Pi(D)\otimes \Pi(F)$, and, for all $\cP_D\in \Pi(D)$ and $\cP_F\in\Pi(F)$, we have $\aux_H(\cP_D\cup\cP_F)=\aux_H(\cP_D)\uplus\aux_H(\cP_F)$.
		Suppose that $\cA_{D}\lesssim_D \Pi(D)$ and $\cA_{F}\lesssim_F \Pi(F)$.
		
		Let $\cP_D\in\Pi(D)$ and $\cP_F\in\Pi(F)$, and let $\cM$ be a multigraph whose edges are colored blue such that there exists a red-blue Eulerian trail $T$ in $\aux_H(\cP_D\cup \cP_F)\uplus\cM$. 
		It is sufficient to prove that there exist $\cP_D^\star\in\cA_D$ and $\cP_F^\star\in\cA_F$ such that $\aux_H(\cP_D^\star\cup\cP_F^\star)\uplus\cM$ admits a red-blue Eulerian trail.
		
		We begin by proving that there exists $\cP_D^\star\in\cA_D$ such that $\aux_H(\cP_D^\star\cup \cP_F)\uplus\cM$ contains a red-blue Eulerian trail.
		In order to do so, we construct from $\aux_H(\cP_D\cup\cP_F)\uplus\cM$ a multigraph $\cM'$ by successively repeating the following: take a maximal sub-walk $W$ of $T$ which uses alternately blue edges and red edges from $\aux_H(\cP_F)\uplus \cM$, remove these edges and add a blue edge between the two end-vertices of $W$.
		
		By construction of $\cM'$, for every $\cP_D'\in\Pi(D)$, if $\aux_D(\cP_D')\uplus\cM'$ admits a red-blue Eulerian trail, then $\aux_H(\cP_D'\cup\cP_F)\uplus\cM$ contains a red-blue Eulerian trail also.
		We also deduce from this construction that the multigraph $\aux_D(\cP_D)\uplus\cM'$ contains a red-blue Eulerian trail.
		As $\cA_D\lesssim_D \Pi(D)$, there exists $\cP_D^\star$ such that $\aux_D(\cP_D^\star)\uplus\cM'$ contains a red-blue Eulerian trail.
		We conclude that $\aux_H(\cP_D^\star\cup\cP_F)\uplus\cM$ contains also a red-blue Eulerian trail.
		
		Symmetrically, we can prove that there exists $\cP_F^\star\in\cA_F$ such that $\aux_H(\cP_D^\star\cup\cP_F^\star)\uplus\cM$ contains a red-blue Eulerian trail. This proves the lemma.
	\end{proof}

	For two $k$-labeled subgraphs $H$ and $D$ arising in the $k$-expression of $G$ such that $H=\eta_{i,j}(D)$, we denote by $\cE^H_{i,j}$ the set of edges whose endpoints are labeled $i$ and $j$, \ie,  $\{ uv\in E(H) \mid  \lab_H(v)=i\ \wedge\ \lab_H(v)=j \}$.
	For $\cP\in \Pi(H)$, we denote by $\cP+(i,j)$ the set of all partial solutions $\cP'$ of $\Pi(H)$ obtained by the union of $\cP$ and an edge $uv$ in $\cE^H_{i,j}$. Observe that $u$ and $v$ must be the endpoints of two distinct maximal paths of $H_{\mid\cP}$.
	We extend this notation to sets of partial solutions; for every $\cA\subseteq \Pi(H)$, we denote by $\cA+(i,j)$, the set $\bigcup_{\cP\in\cA}(\cP+(i,j))$. It is worth noting that $\Pi(D)\subseteq \Pi(H)$ and thus the operator $+(i,j)$ is well defined for a partial solution in $\Pi(D)$.
	
	Moreover, for every $\cA\subseteq \Pi(D)$ and integer $t\geq 0$,  we define the set $\cA^t$ as follows
	\[ \cA^t:=\begin{cases}
	\cA & \text{ if $t=0$,}\\
	\reduce_H(\cA^{t-1}+(i,j)) & \text{ otherwise.}
	\end{cases} \]
	Observe that each set $\cP$ in $\cA^t$ is a partial solution of $H$ and $|\cP\cap \cE_{i,j}^H|=t$.
	
	\begin{lemma}\label{lem:addeges}
		Let $H=\eta_{i,j}(D)$ such that $E(D)\cap \cE^H_{i,j}=\emptyset$.
		If $\cA_D\lesssim_D \Pi(D)$, then we have $\cA_D^0\cup\dots\cup\cA_D^{n}\lesssim_H \Pi(H)$.
	\end{lemma}
	\begin{proof}
	Suppose that $\cA_D\lesssim_D \Pi(D)$.
	We begin by proving the following claim.
	\begin{claim}\label{claim:addeges}
		Let $\cA,\cB\subseteq \Pi(H)$.
		If $\cA\lesssim_H\cB$, then $\cA+(i,j)\lesssim_H \cB+(i,j)$.
	\end{claim}
	\begin{proof}
		Suppose that $\cA\lesssim_H \cB$.
		Let $\cP\in \cB+(i,j)$ and $\cM$ be a multigraph with blue edges such that $\aux_H(\cP)\uplus\cM$ admits a red-blue Eulerian trail.
		Take $xy\in \cE_{i,j}^H$ such that $\cP':=\cP- xy$ belongs to $\cB$ and $x\in \lab_H^{-1}(i)$ and $y\in \lab_H^{-1}(j)$.
		Let $\cM'$ be the multigraph obtained by adding a blue edge $f$ with endpoints $v_i$ and $v_j$ to $\cM$.
		
		We claim that the multigraph $\aux_H(\cP')\uplus\cM'$ admits a red-blue Eulerian trail.
		Note that there is a path $P\in \cP$ containing the edge $xy$, 
		and when we remove $xy$ from $\cP$, we divide $P$ into two maximal subpaths, say $P_1$ and $P_2$.
		Without loss of generality, we may assume that $P_1$ contains $x$ and $P_2$ contains $y$.
		Let $x'$ and $y'$ be the other end-vertices of $P_1$ and $P_2$, respectively, 
		and let $i':=\lab_H(x')$ and $j':=\lab_H(y')$.
		Note that $\aux_H(\cP')$ can be obtained from $\aux_H(\cP)$ by removing an edge $e$ associated with $\{v_{i'},v_{j'}\}$ and adding two edges $e_1$ and $e_2$ associated with $\{v_{i'},v_{i}\}$ and $\{v_j,v_{j'}\}$ respectively.
		So, we can obtain a red-blue Eulerian trail of $\aux_H(\cP')\uplus\cM'$
		from a red-blue Eulerian trail of $\aux_H(\cP)\uplus \cM$ 
		by replacing $(v_{i'},e,v_{j'})$ with the sequence $(v_{i'}, e_1, v_i, f, v_j, e_2, v_{j'})$ 
		where $f$ is the blue edge we add to $\cM$ to obtain $\cM'$.
		It implies the claim.

		Now, since $\cA\lesssim_H \cB$, there exists $\cP^\star\in\cA$ such that $\aux_H(\cP^\star)\uplus\cM'$ admits a red-blue Eulerian trail $T$.
		Let $W$ be the subwalk of $T$ such that $W=(v_a,e_a,v_i,f,v_j,e_b,v_b)$. 
		Take two maximal paths $P_1$ and $P_2$ in $H_{|\cP^\star}$ such that the end-vertices of $P_1$ (respectively $P_2$) are labeled $a$ and $i$ (resp. $b$ and $j$).
		Let $\widehat{\cP}$ be the partial solution of $H$ obtained from $\cP^\star$ by adding the edge in $\cE_{i,j}^H$ between the end-vertex of $P_1$ labeled $i$ and the end-vertex of $P_2$ labeled $j$.
		By construction, we have $\widehat{\cP}\in \cA+ (i,j)$ and $\aux_H(\widehat{\cP})\uplus\cM$ admits a red-blue Eulerian trail.
		We conclude that $\cA+(i,j)\lesssim_H \cB+(i,j)$.
	\end{proof}
	Let $\Pi(D)+t(i,j)$ be the set of partial solutions of $H$ obtained by applying $t$ times the operation $+(i,j)$ on $\Pi(D)$.
	Since every partial solution of $H$ is obtained from the union of a partial solution of $D$ and a subset of $\cE_{i,j}^H$ of size at most $n$, we have $\Pi(H)=\bigcup_{0\leq t \leq n} (\Pi(D)+t(i,j))$.

	Since $V(D)=V(H)$ and $E(D)\subseteq E(H)$, we have $\cA_D^0=\cA_D\lesssim_H \Pi(D)+0(i,j)$.
	Let $\ell \in \{ 1,\dots,n\}$ and suppose that $\cA_D^{\ell-1}\lesssim_H \Pi(D)+(\ell-1)(i,j)$.
	From Claim \ref{claim:addeges}, we have $\cA_D^{\ell-1}+(i,j)\lesssim_H \Pi(D)+\ell(i,j)$.
	By Lemma \ref{lem:reduce}, we deduce that $\cA_D^{\ell}=\reduce(\cA_D^{\ell-1} +(i,j))\lesssim_H \Pi(D)+\ell(i,j)$.
	
	Thus, by recurrence, for every $0\leq\ell\leq n$, we have $\cA_D^\ell \lesssim_H \Pi(D)+\ell(i,j)$.
	We conclude that $\cA_D^0\cup \dots\cup \cA_D^n\lesssim_H \Pi(H)$. 
	\end{proof}
	
	We are now ready to prove the main theorem of this paper.
	
	\begin{theorem}\label{thm:hamiltonian}
		There exists an algorithm that, given an $n$-vertex graph $G$ and a $k$-expression $\phi$ of $G$, solves \textsc{Hamiltonian Cycle} in time $ O(n^{2k+5}\cdot 2^{2k (\log_2(k)+ 1)} \cdot k^3 \cdot \log_2(n k))$.
	\end{theorem}
	\begin{proof}
	 	Since every $k$-expression can be transformed into an irredundant $k$-expression in linear time, we may assume that $\phi$ is an irredundant $k$-expression.
	 	
	 	We do a bottom-up traversal of the $k$-expression and at each $k$-labeled graph $H$ arising in $\phi$, we compute a set $\cA_H\subseteq \Pi(H)$ such that $|\cA_H|\leq n^{k} 2^{k (\log(k)+ 1)}$ and $\cA_H\lesssim_H\Pi(H)$, by doing the following:
	 	\begin{itemize}
	 		\item If $H=i(v)$, then we have $\Pi(H)=\{\emptyset\}$ because $E(H)=\emptyset$. In this case, we set $\cA_H:=\{ \emptyset \}$. Obviously, we have $\cA_H\lesssim_H \Pi(H)$.
	 		\item If $H=\rho_{i,j}(D)$, then we set $\cA_H:=\cA_D$.
	 		\item If $H=D\oplus F$, then we set $\cA_H:=\reduce_H(\cA_D\otimes \cA_F)$. 
	 		\item If $H=\eta_{i,j}(D)$, then we set $\cA_H:=\reduce_H(\cA_D^0\cup \dots\cup \cA_D^n)$.
	 	\end{itemize}
		For the last three cases, we deduce, by induction, from Lemma \ref{lem:reduce} and Lemmas \ref{lem:relabeling}--\ref{lem:addeges} that $\cA_H\lesssim_H \Pi(H)$. Moreover by the use of the function $\reduce_H$, by Lemma \ref{lem:nbequivalence}, we have $|\cA_H|\leq n^{k}\cdot 2^{k (\log(k)+ 1)}$.
		
		\medskip
		We now explain how our algorithm decides whether $G$ admits a Hamiltonian cycle.
		We claim that $G$ has a Hamiltonian cycle if and only if there exist two $k$-labeled graphs $H$ and $D$ arising in $\phi$ with $V(H)=V(G)$ and $H=\eta_{i,j}(D)$, and there exists $\cP\in \cA_D$ such that, for every $\ell \in [k]\setminus\{i,j\}$, we have $\deg_{\aux_D(\cP)}(v_\ell)=0$ and $\deg_{\aux_H(\cP)}(v_i)=\deg_{\aux_H(\cP)}(v_j)>0$.

		First suppose that $G$ contains a Hamiltonian cycle $C$ and take the $k$-labeled graph $H$ arising in $\phi$ such that $E(C)\subseteq E(H)$ and $|E(H)|$ is minimal.
		Note that the operations of taking the disjoint union of two graphs or relabeling cannot create a Hamiltonian cycle. 
		Thus, by minimality, we have $H=\eta_{i,j}(D)$ such that 
		\begin{itemize}
		\item $D$ is a $k$-labeled graph arising in $\phi$ and $i,j\in [k]$, 
		\item $E(C)\not\subseteq  E(D)$.
		\end{itemize}
		It follows that $E(C)\setminus E(D)\subseteq \cE_{i,j}^H$.
		Let $\cP:=E(C)\cap E(D)$ and $\cQ:=E(C)\cap \cE_{i,j}^H$.
		Observe that $\cP\in \Pi(D)$ and $\cQ\in\overline{\Pi}(D)$.
		By Lemma \ref{lem:easyway}, the multigraph $\aux_D(\cP)\uplus\aux_D(\cQ)$ contains a red-blue Eulerian trail.
		Since $\cA_D\lesssim_D \Pi(D)$, there exists $\cP^\star\in\cA_D$ such that $\aux_D(\cP^\star)\uplus\aux_D(\cQ)$ contains a red-blue Eulerian trail.
		As $\cQ\subseteq \cE_{i,j}^H$, we deduce that, for every $\ell \in  [k]\setminus\{i,j\}$, we have $\deg_{\aux_D(\cP^\star)}(v_\ell)=0$ and $\deg_{\aux_H(\cP^\star)}(v_i)=\deg_{\aux_H(\cP^\star)}(v_j)$.
		
		For the other direction, suppose that the latter condition holds.
		Let $\cQ$ be the graph on the vertex set $V(G)$ such that it contains exactly $\deg_{\aux_H(\cP)}(v_i)$ many edges between 
		the set of vertices labeled $i$ and the set of vertices labeled $j$.
		Observe that $\aux_H(\cQ)$ consists of $\deg_{\aux_H(\cP)}(v_i)$ many edges between $v_i$ and $v_j$.
		Therefore, by Lemma~\ref{lem:euleriantrail}, 
		$\aux_H(\cP)\uplus \aux_H(\cQ)$ admits a red-blue Eulerian trail.
		Then, by Lemma~\ref{lem:EulerianHamiltonian}, 
		there exists $\cQ^\star\in \overline{\Pi}(H)$ such that 
		$\cP$ and $\cQ^\star$ form a Hamiltonian cycle, as required.

		\medskip
		\paragraph{\bf Running time} 
		Let $H$ be a $k$-labeled graph arising in $\phi$.
		Observe that if $H=i(v)$ or $H=\rho_{i\rightarrow j}(D)$, then we compute $\cA_H$ in time $O(1)$.
		
		By Lemma \ref{lem:nbequivalence}, for every $\cA\subseteq \Pi(H)$, we can compute $\reduce_H(\cA)$ in time $O(|\cA| \cdot n k^2 \log_2(nk))$.
		Observe that, for every $k$-labeled graph $D$ arising in $\phi$ and such that $\cA_D$ is computed before $\cA_H$, we have $|\cA_D|\leq n^{k} \cdot 2^{k (\log_2(k)+ 1)}$. It follows that:
		\begin{itemize}
			\item If $H=D\oplus F$, then we have $|\cA_D\otimes\cA_F| \leq n^{2k} \cdot 2^{2k (\log_2(k)+ 1)}$. Thus, we can compute $\cA_H:=\reduce_H(\cA_D\otimes\cA_F)$ in time \[ O(n^{2k+1} \cdot 2^{2k (\log_2(k)+ 1)} \cdot k^2 \log_2(n k)). \]
			
			\item If $H=\eta_{i,j}(D)$, then we can compute $\cA_H$ in time 
			\[ O(n^{k+4}\cdot 2^{k(\log_2(k)+1)} \cdot k^2  \log_2(n k)).  \]
			First observe that, for every partial solution $\cP$ of $H$, we have $|\cP+(i,j)|\leq n ^2$ and we can compute the set $\cP+(i,j)$ in time $O(n^2)$.
			Moreover, by Lemma \ref{lem:nbequivalence}, for every $\ell \in \{0,\dots,n-1\}$, we have $|\cA_D^{\ell}|\leq n^{k}\cdot 2^{k (\log_2(k)+ 1)}$ and thus, we deduce that $|\cA_D^{\ell}+(i,j)|\leq n^{k+2}\cdot 2^{k (\log_2(k)+ 1)}$ and that $\cA_D^{\ell+1}$ can be computed in time $O(n^{k+3}\cdot 2^{k (\log_2(k)+ 1)} \cdot k^2 \log_2(n k) )$.
			Thus, we can compute the sets $\cA_D^1,\dots,\cA_D^n$ in time $O(n^{k+4}\cdot 2^{k (\log(k)+ 1)} \cdot k^2  \log_2(n k) )$.
			The running time to compute $\cA_H$ from $\cA_D$ follows from $|\cA_D^0\cup \dots\cup \cA_D^n|\leq n^{k+1}\cdot 2^{2k (\log_2(k)+ 1)}$.
		\end{itemize}
	Since $\phi$ uses at most $O(n)$ disjoint union operations and $O(nk^2)$ unary operations, we deduce that the total running time of our algorithm is  \[ O(n^{2k+5}\cdot 2^{2k (\log_2(k)+ 1)} \cdot k^4 \log_2(n k)). \qedhere\]
		
	\end{proof}

	\section{Directed Hamiltonian cycle}\label{sec:dhc}
	
	In this section, we explain how to adapt our approach for directed graphs.
	 A \emph{$k$-labeled digraph} is a pair $(G,\lab_G)$ of a digraph $G$  and a function $\lab_G$ from $V(G)$ to $[k]$.
	Directed clique-width is also defined in \cite{CourcelleER93}, 
	and it is based on the four operations, where 
	the three operations of creating a digraph, taking the disjoint union of two labeled digraphs, and
	relabeling a digraph are the same, 
	and the operation of adding edges is replaced with the following:
	\begin{itemize}
	\item For a labeled digraph $G$ and distinct labels $i,j\in [k]$, add all the non-existent arcs from vertices with label $i$ to vertices with label $j$ (so we do not add arcs of both directions altogether).
	\end{itemize}
	A \emph{directed clique-width $k$-expression} and \emph{directed clique-width} are defined in the same way.	
	A \emph{directed Hamiltonian cycle} of a digraph $G$ is a directed cycle containing all the vertices of $G$.
	The \textsc{Directed Hamiltonian Cycle} problem asks, for a given digraph $G$, whether $G$ contains a directed Hamiltonian cycle or not.
	
	With a similar approach, we can show the following.
	
	\begin{theorem}
		There exists an algorithm that, given an $n$-vertex digraph $G$ and a directed clique-width $k$-expression of $G$, 
		solves \textsc{Directed Hamiltonian Cycle} in time $n^{O(k)}$.
	\end{theorem}

	For \textsc{Directed Hamiltonian Cycle}, auxiliary graphs should be directed graphs.
	We define partial solutions and auxiliary multigraphs similarly, at the exception that a directed maximal path (resp. $H$-path) from a vertex labeled $i$ to a vertex labeled $j$ in a partial solution (resp. complement solution) corresponds to an arc $(v_i,v_j)$ in its auxiliary multigraph.
	
	Let $G$ be an $n$-vertex directed graph and $\phi$ be a directed irredundant $k$-expression of $G$.
	Similar to the proof of Theorem~\ref{thm:hamiltonian}, 
	for every $k$-labeled graph $H$ arising in $\phi$,
	we recursively compute a set $\cA_H$ of small size such that $\cA_H$ represents $\Pi(H)$ which is the set of all partial solutions of $H$.
	
	It is not hard to see that Lemmas \ref{lem:easyway} and \ref{lem:EulerianHamiltonian} hold also in the directed case. That is, we have an equivalence between directed Hamiltonian cycle in graphs and directed red-blue Eulerian trail in multigraphs.
	Thus, to adapt our ideas for undirected graphs, we only need to characterize the directed multigraphs which admit a red-blue Eulerian trail.
	Fleischner~\cite[Theorem VI.17]{Fleischner90} gave such a characterization for directed graphs without loops and multiple arcs, 
	but the proof can easily be adapted for directed multigraphs.
	
	Let $M$ be a directed multigraph whose arcs are colored red or blue, and let $R$ and $B$ be respectively its set of red and blue arcs.
	We denote by $M^*$ the digraph derived from $M$ with $V(M^*):=\{v_{1},v_{2}\mid v\in V(M)\}$ and $E(M^*):=\{(v_1,w_2) \mid (v,w)\in B \} \cup \{ (v_2,w_1) \mid (v,w)\in R\}$. For a digraph $G$ and a vertex $v$ of $G$, we denote by $\deg^+_G(v)$ and $\deg^-_G(v)$, respectively, the outdegree and indegree of $v$ in $G$.

	\begin{lemma}[Fleischner~\cite{Fleischner90}]\label{lem:directedeuleriantrail}
		Let $M$ be a directed multigraph whose arcs are colored red or blue. 
		The following are equivalent.
		\begin{enumerate}
		\item $M$ has a red-blue Eulerian trail.
		\item $M^*$ has an Eulerian trail.
		\item The underlying undirected graph of $M^*$ has at most one connected component containing an edge, and, for each vertex $v$ of $M^*$, $\deg^+_{M^*}(v)=\deg^-_{M^*}(v)$.
		\end{enumerate} 
	\end{lemma}
	
	In (3), the condition that ``for each vertex $v$ of $M^*$, $\deg^+_{M^*}(v)=\deg^-_{M^*}(v)$'' can be translated to that, for each vertex $v$ of $M$, the number of blue incoming arcs is the same as the number of red outgoing arcs, and the number of red incoming arcs is the same as the number of blue outgoing arcs.
	However, an important point is that instead of having that the underlying undirected graph of $M^*$ has at most one connected component containing an edge, the condition that ``the underlying graph of $M$ is connected'' or ``$M$ is strongly connected'' is not sufficient, because the connectedness of $M^*$ depends on the colors of arcs.
	We give an example. Let $G$ be a graph on $\{x,y,z\}$ with red arcs $(x, y), (y, z)$ and blue arcs $(z, y), (y, x)$.
	Even though $G$ is strongly connected, it does not have a red-blue Eulerian trail, and one can check that $M^*$ has two connected components containing an edge.
	
	To decide whether the underlying undirected graph of $M^*$ has at most one connected component containing an edge, 
	multiple arcs are useless. So, it is enough to keep one partial solution $\cP$ for each degree sequence in $\aux_{H}(\cP)$ and for each set $\{\mult_{\aux_H(\cP)}(e) \mid e\in E(\aux_{H}(\cP))\}$.

	Let $\simeq$ be the equivalence relation on $\Pi(H)$ such that $\cP_1\simeq\cP_2$ if the following are satisfied: 
	\begin{itemize}
		\item for every pair $(v,w)$ of vertices in $\{v_1,\dots,v_k\}$, 
		there exists an arc from $v$ to $w$ in $\aux_H(\cP_1)$ if and only if there exists one in $\aux_H(\cP_2)$,
		\item for every vertex $v$ in $\{v_1,\dots,v_k\}$,
		$\deg^+_{\aux_H(\cP_1)}(v)=\deg^+_{\aux_H(\cP_2)}(v)$ and $\deg^-_{\aux_H(\cP_1)}(v)=\deg^-_{\aux_H(\cP_2)}(v)$.
	\end{itemize}
	
	From Lemma~\ref{lem:directedeuleriantrail} and the definition of $M^*$, we deduce the following fact.
	
	\begin{fact}
		Let $\cP_1,\cP_2\in\Pi(H)$. 
		If $\cP_1\simeq \cP_2$, then $\cP_1$ is part of a directed Hamiltonian cycle in $G$ if and only if $\cP_2$ is part of a directed Hamiltonian cycle in $G$.
	\end{fact} 

	From the definition of $\simeq$, one can show that $|\Pi(H)/\simeq|\leq n^{2k}\cdot 2^{k^2}$. 
	Thus we can follow the lines of the proof for undirected graphs, and easily deduce that one can solve \textsc{Directed Hamiltonian Cycle} in time $n^{O(k)}$, where $k$ is the directed clique-width of the given digraph.

	\section{Conclusion}\label{sec:conclusion}
	
	We have proved that, given a $k$-expression, one can solve \textsc{Hamiltonian Cycle} in time $n^{O(k)}$, 
	and also prove a similar variant for directed graphs.

	One major open question related to clique-width is to determine whether it can be approximated within a constant factor.	
	One can bypass this long-standing open problem by using a related parameter called \emph{rank-width}.
	This parameter was introduced by Oum and Seymour~\cite{OumS06} and admits an efficient algorithm to approximate it within a constant factor.
	Moreover, the clique-width of a graph is always bigger than its rank-width.
	On the other hand, rank-width is harder to manipulate than clique-width.
	To the best of our knowledge, there is no optimal algorithm known for basic problems such as \textsc{Vertex Cover} and \textsc{Dominating Set}, where the best algorithms run in time $2^{O(k^2)}\cdot n^{O(1)}$~\cite{BuiXuanTV13} and the best lower bounds state that we can not solve these problems in time $2^{o(k)}\cdot n^{O(1)}$ unless \ETH fails.
	Improving these algorithms or these lower bounds would be the natural way of continuing the works done on clique-width.
	
	\medskip
	
	Recently, Bergougnoux and Kant\'e \cite{BergougnouxK19} proved that the \textsc{Max Cut} problem is solvable in time $n^{O(k)}$ where $k$ is the clique-width of the input graph without assuming that the graph is given with a  $k$-expression.
	For doing so, they used a related parameter called $\mathbb{Q}$-rank-width and the notion of $d$-neighbor equivalence.
	It would be interesting to know whether the same approach can be used for the \textsc{Hamiltonian Cycle} problem.
	
	\medskip
	
	We conclude with one explicit question. A digraph $D$ is an \emph{out-tree} if $D$ is an oriented tree (an undirected tree with orientations on edges) with only one vertex of indegree zero (called the root). 
	The vertices of out-degree zero are called leaves of $D$. The \emph{Min Leaf Out-Branching} problem asks for a given digraph $D$ and an integer $\ell$, 
	whether there is a spanning out-tree of $D$ with at most $\ell$ leaves. This problem generalizes \textsc{Hamiltonian Path} (and also \textsc{Hamiltonian Cycle}) by taking $\ell=1$. 
	Ganian, Hlin{\v{e}}n{\'y}, and Obdr{\v{z}}{\'a}lek~\cite{GanianHO11} showed that there is an $n^{2^{O(k)}}$-time algorithm for solving \textsc{Min Leaf Out-Branching} problem, 
	when a clique-width $k$-expression of a digraph $D$ is given. We ask whether it is possible to drop down the exponential blow-up from $2^{O(k)}$ to $O(k)$.
	
	\medskip
	
	\paragraph{\textbf{Acknowledgement.}}
The authors would like to thank the anonymous reviewer for pointing out the previous mistake on red-blue Eulerian trails for directed graphs and for indicating proper citations.

\end{document}